\documentclass[conference]{IEEEtran}
\IEEEoverridecommandlockouts
\usepackage{cite}
\usepackage{graphicx}
\usepackage{amsmath,amssymb,amsfonts}
\usepackage[center]{caption}
\usepackage[font=footnotesize]{caption}
\usepackage{xcolor}
\usepackage{algorithm,algorithmicx,algpseudocode}
\usepackage{bm}
\usepackage{mathtools}
\usepackage{nccmath}
\usepackage{enumitem}

\usepackage{amsthm}
\newtheorem{theorem}{Theorem}
\newtheorem{proposition}{Proposition}
\newtheorem{definition}{Definition}

\DeclareMathOperator*{\argmin}{arg\,min}

\title{\LARGE Optimizing Information Freshness in IoT Systems with Update Rate Constraints: A Token-Based Approach}

\author{
	\IEEEauthorblockN{Erfan Delfani\IEEEauthorrefmark{1}, and Nikolaos Pappas\IEEEauthorrefmark{1}\\} 
	\IEEEauthorblockA{\IEEEauthorrefmark{1}Department of Computer and Information Science, Link\"{o}ping University, Link\"{o}ping, Sweden}
	Email: \{erfan.delfani, nikolaos.pappas\}@liu.se
	\thanks{This work has been supported in part by the Swedish Research Council (VR), ELLIIT, and the European Union (ETHER, 101096526, ELIXIRION, 101120135, and SOVEREIGN, 101131481).}
	\thanks{ISBN 978-3-903176-63-8 $\copyright$ 2024 IFIP}
}

\begin{document}
	
	\setlength{\abovecaptionskip}{1pt}
	\setlength{\belowcaptionskip}{-10pt} 
	\setlength{\textfloatsep}{0pt}
	
	\maketitle
	
	\begin{abstract}
		In Internet of Things (IoT) status update systems, where information is sampled and subsequently transmitted from a source to a destination node, the imperative necessity lies in maintaining the timeliness of information and updating the system with optimal frequency. Optimizing information freshness in resource-limited status update systems often involves \emph{Constrained Markov Decision Process (CMDP)} problems with \emph{update rate constraints}. Solving CMDP problems, especially with multiple constraints, is a challenging task. To address this, we present a \emph{token-based approach} that transforms CMDP into an unconstrained MDP, simplifying the solution process. We apply this approach to systems with one and two update rate constraints for optimizing \emph{Age of Incorrect Information (AoII)} and \emph{Age of Information (AoI)} metrics, respectively, and explore the analytical and numerical aspects. Additionally, we introduce an \emph{iterative triangle bisection} method for solving the CMDP problems with two constraints, comparing its results with the token-based MDP approach. Our findings show that the token-based approach yields superior performance over baseline policies, converging to the optimal policy as the maximum number of tokens increases.
	\end{abstract}
	
	\section{Introduction}
	Prudent control of status updates to provide timely information exchange in resource-limited IoT networks has emerged as an important research direction within semantics-aware goal-oriented communication. In these systems, \emph{update packets} from an information source are generated and transmitted through a communication network to a destination node for subsequent processing and utilization. This process fulfills application-driven goals, including remote actuator control in industrial automation, steering autonomous vehicles in intelligent transportation systems, and monitoring health and environmental status\cite{kaul2012real,abd2019role, pappas2023age}.
	In such time-sensitive systems, the reliability and accuracy of controls depend on the timeliness of information. Therefore, there is a high demand for the \emph{freshness of information} and, consecutively, for frequent generation of updates. However, due to resource limitations arising from factors such as energy limitations, channel accessibility, or processing capacity, it is crucial to optimize the timing and the amount of generated, transmitted, or processed data within a network. This results in problems addressing the optimization of information freshness under updating rate constraints, aiming to find \emph{optimal policies} that answer the question of \emph{when} to generate or transmit update packets. In the literature of status update systems, the quantification of information freshness is accomplished through timing metrics, such as Age of Information (AoI) \cite{kaul2012real}, Age of Incorrect Information (AoII) \cite{maatouk2020age}, Version AoI (VAoI) \cite{yates2021age}, and Age of Actuation (AoA) \cite{nikkhah2023age}, each metric addressing distinct aspects of timing and importance of information.
	
	Below, we present and discuss papers that consider the optimization of information freshness in status update systems under rate constraints. \cite{ceran2019average} minimizes average AoI by scheduling updates over an error-prone channel with transmission constraints, studying optimal policies under various feedback mechanisms. \cite{wang2019minimizing} addresses minimizing average AoI under transmission power constraints, transforming the problem into an unconstrained MDP through Lagrangian relaxation. \cite{maatouk2020age} introduces AoII metric and solves the CMDP problem with a Lagrangian approach. \cite{tang2020minimizing} focuses on fresh data collection from power-constrained sensors in IIoT networks, optimizing average AoI with scheduling algorithms and LP-based optimization. \cite{chen2023minimizing} minimizes AoI in wireless networks with peak power-constrained base stations, using CMDP with strict and relaxed power constraints and proposing a truncated multi-user scheduling policy. \cite{hatami2022demand} minimizes average AoI in resource-constrained IoT networks using a CMDP model and Lagrangian multipliers, offering an asymptotically optimal, low-complexity algorithm. 
	The works mentioned above primarily rely on formulating the constrained optimization problem using the CMDP framework and subsequently proceed to solve it through a Lagrangian approach.
	
	However, solving a CMDP problem through the direct primary formulation or the Lagrangian dual approach presents a formidable challenge, particularly for problems with multiple constraints\cite{altman1999constrained}. To address this challenge, we present a \emph{token-based approach} for transforming the CMDP problem into an unconstrained MDP problem, following the same approach in \cite{stamatakis2023optimizing}. The resulting MDP can be directly solved using popular iterative standard methods. Specifically, we convert the constrained problem into an unconstrained problem by defining new variables within the system model. These variables are defined in such a way as to ensure the update rate constraints. To elaborate, drawing inspiration from the well-known \emph{token bucket} mechanism, a commonly referenced concept in the literature \cite{tanenbaum2011computer,raghunathan2004energy,chiariotti2022query,stamatakis2023optimizing}, we allocate a specific number of tokens to the system in each time slot (or assign one token with a specific probability) for potential updates. When the tokens accumulate sufficiently, the system can decide whether to spend them and update or refrain from updating. This way, the system can optimize update actions by spending tokens at the right time by solving an unconstrained MDP.
	For single rate constraint problems, this approach resembles the idea of modifying the system model and incorporating an Energy Harvesting (EH) sensor with a random energy arrival for sampling and status updating \cite{yates2015lazy,bacinoglu2017scheduling,stamatakis2019control,chen2021optimizing,gindullina2021age,hatami2021aoi,delfani2023version}. However, as we will demonstrate in subsequent sections by applying it to two rate constrained problems, this token-based approach can be effectively applied to general system models, without being limited to EH scenarios. Our contributions are as follows:
	\begin{itemize}[leftmargin=0.12in]
		\item We present a token-based unconstrained MDP for optimizing AoII in a status update system derived from a single-rate constrained MDP and provide analytical results.
		\item We formulate a two-rate constrained MDP and its corresponding token-based unconstrained MDP for optimizing AoI in a status update system. Analytical results for the unconstrained problem are presented.
		\item We introduce an iterative triangle bisection method for solving the CMDP with two constraints.
		\item Finally, we compare the performance of the optimal token-based policy in different setups with that of the optimal policy for the primary CMDP and other baseline policies through numerical analysis.
	\end{itemize}
	
	\section{System Model}
	\label{SystemModelSection}
	We consider a status update system depicted in Fig. \ref{fig:GeneralSystemModel}, where update packets/samples originating from an information source are generated and transmitted from a source node (Tx) to a destination node (Rx) within a network, by an update policy. This policy dictates when the update packets are transmitted to the destination nodes, and our objective is to ascertain the optimal policy that maximizes performance by optimizing information freshness within the system while adhering to maximum allowable update rates. In this system model, the update packets are forwarded through a channel, which may be reliable or unreliable. Additionally, we integrate a channel from the receiver to the transmitter, which can be utilized to either transmit acknowledgment feedback (in a \emph{push-based} scenario) or request a new update (in a \emph{pull-based} scenario), as in \cite{agheli2023effective} and \cite{stamatakis2024semantics}. We presume the backward channel to be error-free and instantaneous. The proposed token-based approach is versatile and can be applied to optimize various information freshness metrics across diverse system configurations. In what follows, we investigate the CMDP and token-based problems to optimize AoII and AoI in two distinct system setups under one and two update rate constraints, respectively, in Sections \ref{SingleRateSection} and \ref{TwoRateSection}. 
	AoI is a performance metric that quantifies the freshness of information, defined as the time elapsed since the generation time of the last successfully received update at the destination node. AoII quantifies the time elapsed since the latest instance when the source collected a sample with identical content to the current sample stored at the destination. In this system model, we assume that the time is slotted and the slots are of equal duration.
	\begin{figure}[t!]
		\centering
		\includegraphics[width=3.4in]{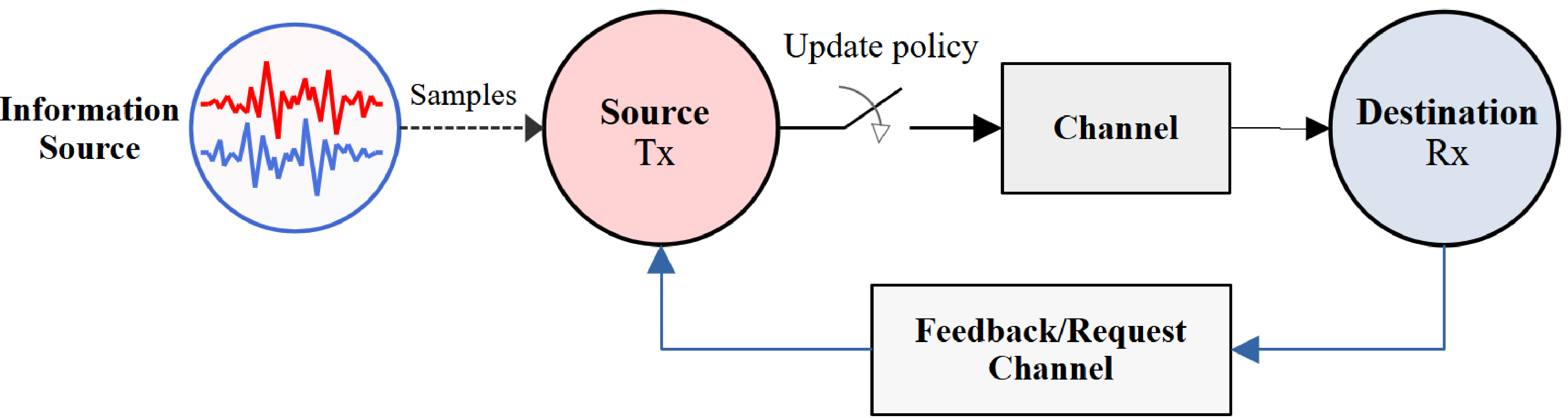}
		\caption{A general system model for optimizing information freshness under update rate constraints.}
		\label{fig:GeneralSystemModel}
		\vspace{10pt}
	\end{figure}
	
	\section{A system with single rate constraint}
	\label{SingleRateSection}
	
	\subsection{Problem Formulation}
	
	In a status update system, as depicted in Fig. \ref{fig:GeneralSystemModel}, where the objective is to optimize the time average of a freshness metric or a \emph{cost function} $C\big(s\left(t\right),a\left(t\right)\big)$ subject to an update rate constraint, a CMDP problem can be defined as follows:
	\begin{align}
		\label{CMDP_eqn}
		\min_{\pi \in \Pi} \ &\limsup_{T\rightarrow\infty} {\frac{1}{T} E\left[ \sum_{t=0}^{T-1} C\big( s^\pi(t),a^\pi(t) \big) \Big| s(0) \right]}, \\
		\text{s.t.} \ &\limsup_{T\rightarrow\infty} {\frac{1}{T} E\left[ \sum_{t=0}^{T-1} a^\pi(t) \Big| s(0) \right]} \leq \alpha, \notag
	\end{align}
	where the \emph{policy} $\pi$ is a sequence of \emph{actions}, i.e., $\pi\overset{\text{def}}{=}\big(a^\pi(0),a^\pi(1),a^\pi(2),\cdots\big)$, where $a^\pi(t) \in \{0,1\}$ represents the action at time $t$. Here, $a^\pi(t)=1$ represents the \emph{update} action, whereas $a^\pi(t)=0$ indicates the \emph{no-update} action. The term \emph{update} refers to the process of generating and transmitting information from the source to the destination. $s^\pi(t)$ is the \emph{state} of the system under the policy $\pi$, including a timeliness or freshness variable. In this formulation, the parameter $0 < \alpha < 1$ restricts the average update rate. We omit the superscript $\pi$ in our analysis to facilitate presentation.
	
	To transform this CMDP problem into the token-based MDP one, we define a \emph{new state vector} $\hat{s}(t)\overset{\text{def}}{=}[b(t),s(t)]$, where $b(t) \in \{0,1,2,\ldots,b_{max}\}$ is a token variable evolving as:
	\begin{align}
		\label{bEvol_eqn}
		b(t+1)=
		\begin{cases}
			b(t) - a(t) + 1 & \text{with probability } \alpha, \\
			b(t) - a(t) & \text{with probability } 1-\alpha, \\
		\end{cases}
	\end{align}
	where we have assumed that each update consumes one token. The action $a(t)$ must be set to zero (indicating no update) when $b(t)$ is equal to zero. Otherwise, when $b(t)>0$, the system will decide on the optimal action at each time slot. This optimal policy can be obtained by solving the following \emph{token-based unconstrained MDP} problem:
	\begin{align}
		\label{UncMDP_eqn}
		\min_{\pi \in \Pi}\ \limsup_{T\rightarrow\infty} {\frac{1}{T} E\Bigg[ \sum_{t=0}^{T-1} C\Big( \hat{s}^\pi(t),a^\pi(t) \Big) \Big| \hat{s}(0) \Bigg]},
	\end{align}
	where the cost function is the same as the CMDP problem \eqref{CMDP_eqn}, and the transition probabilities of this new MDP problem can be simply derived using the transition probability of the main CMDP problem, via the following equation:
	\begin{multline}
		\label{eqn:NewStateVecTrans}
		P \big[\hat{s}(t\!+\!1) \big| \hat{s}(t),a(t) \big]\!=\!P \big[b(t\!+\!1), s(t\!+\!1) \big| b(t), s(t),a(t) \big]  \\
		= \underbrace{P \big[b(t+1) \big| b(t),a(t) \big]}_{\text{Can be derived from \eqref{bEvol_eqn}}} \times \underbrace{P \big[s(t+1) \big| s(t),a(t) \big]}_{\text{Given by CMDP problem \eqref{CMDP_eqn}}}.
	\end{multline}
	
	\subsection{Optimizing AoII under rate constraint} 
	In a scenario where the cost function is defined to be AoII, denoted by $\Delta^{\text{AoII}}(t)$, the CMDP problem \eqref{CMDP_eqn} has previously been addressed for a specific system model \cite{maatouk2020age}. The considered system model focuses on a communication system involving a transmitter-receiver pair where the transmitter sends status updates about a discrete Markov chain process to the receiver over an unreliable channel, as depicted in Fig. \ref{fig:GeneralSystemModel}. The process has $N$ states, and probabilities of staying in the same state ($p_R$) or transitioning to another state ($p_t$) are defined such that $p_R+(N-1)p_t=1$. The unreliable channel follows a Bernoulli distribution with success ($p_s$) and failure ($p_f=1-p_s$) probabilities at each time slot. Successful transmission prompts an ACK feedback, while failure results in a NACK. The transmitter can perfectly estimate the receiver's information source using these packets. The transmitter follows a push-based scenario and generates updates at its discretion by sampling the current process state. The objective is to minimize the time-average AoII with the adopted transmission policy. 
	
	This problem can be formulated as a CMDP presented in \eqref{CMDP_eqn}, where the \emph{state} at time $t$ is characterized by the AoII, $s(t)\overset{\text{def}}{=}\Delta^{\text{AoII}}(t) \in \{0,1,2,\cdots,\Delta_{\text{max}}\}$. The \emph{action} at time $t$, denoted as $a(t) \in \{0,1\}$, signifies an attempted update (value $1$) or remaining idle (value $0$). The \emph{instantaneous cost} is defined as $C\left( s(t),a(t) \right)\overset{\text{def}}{=}\Delta^{\text{AoII}}(t)$, and the \emph{transition probabilities} between states are detailed below:
	
	\begin{subequations}
		\label{eqn:TransitionsAoII}
		\begin{align}
			P&\left[s(t+1)\!=\!0 \big| s(t)\!=\!0,a(t)\right]\!=\!p_R,\ \forall a(t) \in \{0,1\} \\
			P&\left[s(t+1)\!=\!1 \big| s(t)\!=\!0,a(t)\right]\!=\!1\!-\!p_R,\ \forall a(t) \in \{0,1\} \\
			P&\left[s(t+1)\!=\!0 \big| s(t)\neq 0, a(t)\!=\!0\right] \!=\! p_t \\
			P&\left[s(t+1)\!=\! s(t) \!+\! 1 \big| s(t) \neq 0, a(t) \!=\! 0\right] \!=\! 1 \!-\! p_t \\
			P&\left[s(t+1)\!=\!0 \big| s(t)\neq0, a(t)\!=\!1\right]\!=\! \beta \\
			P&\left[s(t+1) \!=\! s(t) \!+\! 1 \big| s(t) \neq 0, a(t) \!=\! 1\right] \!=\! 1 \!-\! \beta,
		\end{align}
	\end{subequations}
	\noindent where $\beta \overset{\text{def}}{=} p_Rp_s + p_fp_t$. For the main problem presented in \cite{maatouk2020age}, the assumption $p_R>p_t$ has been considered to prevent the trivial optimal scenario of \emph{never update}. We consider the same assumption which results in:
	$\beta\!=\!p_Rp_s \!+\! p_fp_t \!>\! p_tp_s \!+\! p_fp_t \!=\! p_t$.
	
	Transforming this CMDP problem into a token-based MDP problem, the transition probabilities for the new state vector $\hat{s}$, i.e., $P\left[\hat{s}^\prime|\hat{s},a\right]=P\left[b^\prime,s^\prime \big| b,s,a\right]$, can now be calculated using \eqref{eqn:NewStateVecTrans}, \eqref{bEvol_eqn} and \eqref{eqn:TransitionsAoII}. 
	
	\begin{itemize}[leftmargin=0.12in]
		\item \textit{Case 1.} $a = 0$ and $s = 0$.
		\begin{align}
			P\left[\hat{s}^\prime|\hat{s},a\right] \!=\!
			\begin{cases}
				\alpha p_R & b^\prime\!=\!b\!+\!1, s^\prime \!=\! 0, \\
				\alpha (1\!-\!p_R) & b^\prime\!=\!b\!+\!1, s^\prime \!=\! 1, \\
				(1\!-\!\alpha)p_R & b^\prime\!=\!b, s^\prime \!=\! 0, \\
				(1\!-\!\alpha)(1\!-\!p_R) & b^\prime\!=\!b, s^\prime \!=\! 1.
			\end{cases} 
		\end{align}
		
		\item \textit{Case 2.} $a = 0$ and $s \neq 0$.
		\begin{align}
			P&\left[\hat{s}^\prime|\hat{s},a\right] \!=\!
			\begin{cases}
				\alpha p_t & b^\prime\!=\!b\!+\!1, s^\prime \!=\! 0, \\
				\alpha \left(1\!-\!p_t\right) & b^\prime\!=\!b\!+\!1, s^\prime \!=\! s\!+\!1, \\
				(1\!-\!\alpha)p_t & b^\prime\!=\!b, s^\prime \!=\! 0, \\
				(1\!-\!\alpha)\left(1\!-\!p_t\right) & b^\prime\!=\!b, s^\prime \!=\! s\!+\!1.
			\end{cases} 
		\end{align}
		
		\item \textit{Case 3.} $a = 1$ and $s = 0$, while $b>0$.
		\begin{align}
			P&\left[\hat{s}^\prime|\hat{s},a\right] \!=\!
			\begin{cases}
				\alpha p_R & b^\prime\!=\!b, s^\prime \!=\! 0, \\
				\alpha (1\!-\!p_R) & b^\prime\!=\!b, s^\prime \!=\! 1, \\
				(1\!-\!\alpha)p_R & b^\prime\!=\!b\!-\!1, s^\prime \!=\! 0, \\
				(1\!-\!\alpha)(1\!-\!p_R) & b^\prime\!=\!b\!-\!1, s^\prime \!=\! 1.
			\end{cases} 
		\end{align}
		
		\item \textit{Case 4.} $a = 1$ and $s \neq 0$, while $b>0$.
		\begin{align}
			P&\left[\hat{s}^\prime|\hat{s},a\right] \!=\!
			\begin{cases}
				\alpha \beta & b^\prime\!=\!b, s^\prime \!=\! 0, \\
				\alpha \left(1\!-\!\beta\right) & b^\prime\!=\!b, s^\prime \!=\! s\!+\!1, \\
				(1-\alpha)\beta & b^\prime\!=\!b\!-\!1, s^\prime \!=\! 0, \\
				(1-\alpha)\left(1\!-\!\beta\right) & b^\prime\!=\!b\!-\!1, s^\prime \!=\! s\!+\!1.
			\end{cases} 
		\end{align}
		
	\end{itemize}
	
	In the following, we present the structural results of the optimal token-based policy, and in Section \ref{SingleRateNumericalSection} we compare its outcomes with the main CMDP solution for the same problem.
	
	\subsection{Analytical Results} 
	We present analytical results regarding the existence and structure of the optimal token-based policy. Here, we omit the \emph{hat} on $\hat{s}$ for the sake of simplicity.
	
	\begin{definition}
		An MDP is weakly accessible (or weakly communicating) if its states can be divided into two subsets, $S_t$ and $S_c$, where states in $S_t$ are transient under any stationary policy, and for any two states $s$ and $s^\prime$ in $S_c$, $s^\prime$ can be reached from $s$ under some stationary policy.
	\end{definition}
	
	\begin{proposition}
		The token-based MDP problem \eqref{UncMDP_eqn} is weakly accessible.
	\end{proposition}
	
	\begin{proof}
		We demonstrate that any state $s^\prime=\left(b^\prime,\Delta^{\text{AoII}^\prime}\right)=\left(b^\prime,\Delta^\prime\right) \in \hat{S}$ is reachable from any other state $s=\left(b,\Delta^{\text{AoII}}\right)=\left(b,\Delta\right) \in \hat{S}$ under a stationary stochastic policy $\pi$, where the action $a\in\{0,1\}$ at each state is randomly selected with a positive probability.
		The state $b^\prime<b$ is accessible from $b$ with a positive probability (w.p.p.) by executing action $a=1$ for $(b-b^\prime)$ time slots, and $b^\prime \geq b$ is reachable from $b$ w.p.p. by realizing action $a=0$ for $(b^\prime-b)$ slots. Upon reaching the state $b^\prime$, irrespective of subsequent actions, the state of the token variable can remain unchanged w.p.p. Consequently, for the remainder of the proof, we consider the token state to be $b^\prime$.
		The state $\Delta^\prime<\Delta$ is also attainable from $\Delta$ w.p.p. by executing action \(a=1\) for one time slot, followed by executing action $a=0$ for $\Delta^\prime$ slots. Meanwhile, the state $\Delta^\prime \geq \Delta$ is accessible from $\Delta$ by executing action $a=0$ for $\Delta^\prime-\Delta$ slots. 
	\end{proof}
	
	\begin{proposition}
		In the token-based MDP problem \eqref{UncMDP_eqn}, the optimal average cost $J^\ast$ achieved by an optimal policy $\pi^\ast$ is the same for all initial states, and it satisfies the Bellman’s equation:
		\vspace{-5pt}
		\begin{equation}
			\label{Bellman_eqn}
			J^\ast\!+\!V(s)\!=\!\min_{a\in\left\{0,1\right\}}{\!\bigg\{\!\sum_{s^\prime\in \hat{S}}{P\left(s^\prime \big | s,a\right)\!\big[C(s,a) \!+\! V(s^\prime)\big]}\!\bigg\}},
			\vspace{-5pt}
		\end{equation}
		\begin{equation}
			\label{OptimalAction_eqn}
			\pi^\ast(s)\! \in \! \argmin_{a\in\left\{0,1\right\}}{\!\bigg\{\!\sum_{s^\prime\in \hat{S}}{P\left(s^\prime \big | s,a\right)\!\big[C(s,a)\!+\!V(s^\prime)\big]}\!\bigg\}}.
		\end{equation}
		where $V(s)$ denotes the value function of the MDP problem.
	\end{proposition}

	\begin{proof}
		As per Proposition 1, the problem defined by equation \eqref{UncMDP_eqn} is weakly accessible. Consequently, in accordance with Proposition 4.2.3 in \cite{bertsekas2011dynamic}, the optimal average cost remains consistent across all initial states. Furthermore, based on Proposition 4.2.6 in \cite{bertsekas2011dynamic}, there exists an optimal policy. According to Proposition 4.2.1 in \cite{bertsekas2011dynamic}, if we can identify $J^\ast$ and $V(s)$ satisfying \eqref{Bellman_eqn}, then the optimal policy can be determined using \eqref{OptimalAction_eqn}.
	\end{proof}
	
	The optimal policy, denoted as $\pi^\ast$, relies on $V(s)$, which typically lacks a solution in closed form. Standard methods, like the (Relative) Value Iteration and Policy Iteration algorithms, can be employed to solve this optimization problem.
	
	\begin{definition}
		Suppose that there exists a $\Delta_{T}(b) > 0$ for each $b$ such that the action $\pi(b,\Delta^{\text{AoII}})$ is $1$ for $\Delta^{\text{AoII}} \geq \Delta_{T}(b)$, and $0$ otherwise. In this case, the policy $\pi$ is a threshold policy.
	\end{definition}
	
	\begin{theorem}
		The optimal policy of the token-based MDP problem \eqref{UncMDP_eqn} is a threshold policy.
	\end{theorem}
	
	\begin{proof}
		The proof can be found in Appendix \ref{Apen1:Theorem1}.
	\end{proof}

	\section{A system with two rate constraints}
	\label{TwoRateSection}
	Consider a scenario where the sensors of an IoT device measure physical variables and transmit updates to a user's monitoring node, such as a smartphone, it is reasonable to assume that a higher demand for more frequent updates or fresher information arises when the user actively monitors the variables on their phone. Conversely, when there is no active monitoring, updates occur at a regular (or minimum) rate. In simple terms, the system imposes two rate constraints in a pull-based scenario: one in response to user requests and another when there is no request. Our objective herein is to obtain an optimal update policy that minimizes the AoI in such a system while simultaneously satisfying the two rate constraints.
	
	\subsection{Problem Formulation}
	Considering the system model in Fig. \eqref{fig:GeneralSystemModel}, the optimization of average AoI under these two constraints can be cast into an infinite-horizon average cost CMDP problem, given by,
	
	\begin{align}
		\label{CMDP_eqn_2rate}
		\min_{\pi \in \Pi} \ &\limsup_{T\rightarrow\infty} {\frac{1}{T} E\Bigg[ \sum_{t=0}^{T-1} C\Big( s^\pi(t),a^\pi(t) \Big) \Big| s(0) \Bigg]}, \\
		\text{s.t.}\ &\limsup_{T\rightarrow\infty} {\frac{1}{T} E\Bigg[ \sum_{t=0}^{T-1} \Big(1-r(t)\Big)a^\pi(t) \Big| s(0) \Bigg]} \leq \alpha_{0}, \notag \\
		\text{s.t.}\ &\limsup_{T\rightarrow\infty} {\frac{1}{T} E\Bigg[ \sum_{t=0}^{T-1} r(t)a^\pi(t) \Big| s(0) \Bigg]} \leq \alpha_{1}, \notag
	\end{align}
	where the \emph{state} is defined as $s(t)\overset{\text{def}}{=}\left[\Delta(t),r(t)\right]^T$, $\Delta(t)$ is the AoI, and $r(t)$ is the random process related to the user's requests; $r(t) = 1$ when there is a user request, and $r(t) = 0$ when there is no request. We assume that the request arrival process follows an i.i.d. Bernoulli distribution with parameter $q$ over time slots.
	Same as the problem \eqref{CMDP_eqn}, $\pi$ is a sequence of \emph{actions} $\pi\overset{\text{def}}{=}\left(a^\pi(0),a^\pi(1),a^\pi(2),\ldots\right)$, $a^\pi(t) \in \{0,1\}$, where $a^\pi(t)=0$ denotes the \emph{remain idle} action, and $a^\pi(t)=1$ denotes the \emph{update} action at time $t$. In this system model, the action is adopted based on the state vector known at the beginning of each time slot, and the receiver is updated through an error-free channel with a lag normalized to one time slot. Finally, the \emph{cost function} is defined to be equal to the AoI, i.e., $C\big( s(t),a(t) \big)\overset{\text{def}}{=}\Delta(t)$.
	In this formulation, $\alpha_{\text{1}}$ and $\alpha_{\text{0}}$ impose limits on the periods of time during which the system is transmitting updates, while simultaneously either receiving or not receiving a request from the receiver, respectively. These two parameters can also be represented as follows:
	\begin{gather}
		\alpha_{0} \!=\! \limsup_{T\rightarrow\infty} {\frac{1}{T} E\left[ \sum_{t=0}^{T-1} \Big(1\!-\!r(t)\Big) \Big| s(0) \right]} \! \times\! \alpha_{\text{min}} \!=\! (1\!-\!q) \alpha_{\text{min}},\notag \\
		\alpha_{1} = \limsup_{T\rightarrow\infty} {\frac{1}{T} E\left[ \sum_{t=0}^{T-1} r(t) \Big| s(0) \right]} \times \alpha_{\text{max}} = q \alpha_{\text{max}}.
	\end{gather}
	
	Here, $\alpha_{\text{max}}$ and $\alpha_{\text{min}}$ denote the maximum desired update rates given that there is or is not a request, respectively.
	
	\subsubsection{Dual Lagrangian Problem}
	The \emph{primary} CMDP problem \eqref{CMDP_eqn_2rate} can be described by introducing two Lagrange multipliers $\lambda_0 \geq 0$ and $\lambda_1 \geq 0$ and defining Lagrangian function as:
	
	{\small
		\begin{align}
			\label{eqn:LagrangeFunc}
			\mathcal{L}&(\bm{\lambda},\pi) \\
			&\overset{\text{def}}{=} \limsup_{T\rightarrow\infty} {\frac{1}{T} E\left[ \sum_{t=0}^{T-1} C\big( s^{\pi}(t),a^{\pi}(t) \big) \Big| s(0) \right]} \!+\! \lambda_0 c_0(\pi) \!+\! \lambda_1 c_1(\pi) \notag \\ 
			&= \limsup_{T\rightarrow\infty} \frac{1}{T} E\Bigg[ \sum_{t=0}^{T-1} C\big( s^{\pi}(t),a^{\pi}(t) \big) + \lambda_0 \big(1-r(t)\big)a^{\pi}(t) \notag \\
			& \qquad \qquad \qquad \qquad + \lambda_1 r(t)a^{\pi}(t) \Big| s(0) \Bigg] - \lambda_0 \alpha_0 - \lambda_1 \alpha_1, \notag
		\end{align}
	}
	
	\noindent where $\bm{\lambda}\overset{\text{def}}{=}[\lambda_0\ \lambda_1]^T$ is the vector of Lagrangian multipliers. Here, we have defined, 
	
	\begin{align}
		c_0(\pi)&\overset{\text{def}}{=}\limsup_{T\rightarrow\infty} {\frac{1}{T} E\left[ \sum_{t=0}^{T-1} \Big(1-r(t)\Big)a^\pi(t) \Big| s(0) \right]} - \alpha_{0} \notag \\ c_1(\pi)&\overset{\text{def}}{=}\limsup_{T\rightarrow\infty} {\frac{1}{T} E\left[ \sum_{t=0}^{T-1} r(t)a^\pi(t) \Big| s(0) \right]} - \alpha_{1},
	\end{align}
	
	\noindent where $c_0(\pi) \leq 0$ and $c_1(\pi) \leq 0$ represent the constraints of the CMDP problem \eqref{CMDP_eqn_2rate}. The Lagrangian \emph{dual} problem then is given by:
	
	\begin{align}
		\label{eqn:DualMDP}
		\sup_{\bm{\lambda} \geq 0} \underbrace{\min_{\pi \in \Pi} \mathcal{L}(\bm{\lambda},\pi)}_{\overset{\text{def}}{=}g(\bm{\lambda})},
	\end{align}
	where $g(\bm{\lambda}) = g(\lambda_0, \lambda_1) = \mathcal{L}(\bm{\lambda},\pi_{\bm{\lambda}}^\ast)$ represents the Lagrange dual function, where $\pi_{\bm{\lambda}}^\ast$ is a $\bm{\lambda}$-optimal policy derived from an unconstrained MDP problem for a given $\bm{\lambda}$, as follows:
	\begin{align}
		\label{eqn:DualMDPpolicy}
		\pi_{\bm{\lambda}}^\ast \in \argmin_{\pi \in \Pi} \mathcal{L}(\bm{\lambda},\pi).
	\end{align}
	
	The new cost function for the dual problem \eqref{eqn:DualMDP}, denoted by $C_{\bm{\lambda}}\big( s(t),a(t) \big)$, can be redefined according to \eqref{eqn:LagrangeFunc}:
	\begin{align}
		\label{LagrangeCostFunc}
		C_{\bm{\lambda}}\big( s(t),a(t) \big) \!\overset{\text{def}}{=}\! C\big( s(t),a(t) \big) \!+\! \lambda_0 \big(1\!-\!r(t)\big)a(t) \!+\! \lambda_1 r(t)a(t).
	\end{align}
	
	The solution to the dual Lagrangian problem \eqref{eqn:DualMDP} generally provides a lower bound for the primary CMDP problem \eqref{CMDP_eqn_2rate}. However, given that the state space of the problem $\mathcal{S}$ is a finite set, the growth condition \cite[Eq. 11.21]{altman1999constrained} is satisfied. Additionally, as the cost function $C\big( s(t),a(t) \big) \geq 0$ is bounded from below, the conditions outlined in \cite[Corollary 12.2]{altman1999constrained} are met, ensuring the equivalence of the optimal solutions for both the dual and primary problems. Consequently, the optimal solution for the primary CMDP \eqref{eqn:DualMDP} can be obtained by solving:
	\begin{align}
		\label{eqn:OptimalDual}
		\sup_{\bm{\lambda} \geq 0} \mathcal{L}(\bm{\lambda},\pi_{\bm{\lambda}}^\ast),
	\end{align}
	where $\pi_{\bm{\lambda}}^\ast$ is derived from \eqref{eqn:DualMDPpolicy}. More specifically, the optimal policy can be determined in two steps: first, by solving the unconstrained MDP \eqref{eqn:DualMDPpolicy} and identifying the $\bm{\lambda}$-optimal policy $\pi_{\bm{\lambda}}^\ast$; second, by obtaining the optimal value of the Lagrangian vector $\bm{\lambda}$ as per \eqref{eqn:OptimalDual}. It is important to note that according to \eqref{LagrangeCostFunc}, the cost of action $a(t)=1$ increases with the increase of $\bm{\lambda}$ multipliers, while it remains constant for action $a(t)=0$. Thus, the cost function $C_{\bm{\lambda}}\big( s(t),a(t) \big)$ will increase while the rates will decrease with the increase of $\bm{\lambda}$. On the other hand, the dual function $g(\bm{\lambda})$ is decreasing in $\bm{\lambda}$ \cite[Lemma 3.1]{beutler1985optimal}. Therefore, the search is for the lowest value of $\bm{\lambda}$ that results in the maximum allowable rates. In Section \ref{IteratibeTriangleSection}, we present an iterative algorithm to obtain the optimal policy and $\bm{\lambda}$. 
	
	\subsubsection{Formulation of the Token-based MDP problem}
	We convert the primary CMDP problem \eqref{CMDP_eqn_2rate} into a token-based MDP problem by including token variables within the state vector, as mentioned earlier. This token-based MDP is characterized by a tuple $<\hat{\mathcal{S}},\mathcal{A},P,C>$, where $\hat{\mathcal{S}}$ is the state space, $\mathcal{A}$ is the set of actions, $P$ is the state transition probability function, and $C$ is the cost of MDP.
	
	\begin{itemize}[leftmargin=0.2in]
		\item States: the new state vector $\hat{s}(t)$ is defined as $\hat{s}(t)\overset{\text{def}}{=}\left[b_0(t),b_1(t),\Delta(t),r(t)\right]^T \in \hat{\mathcal{S}}$, where $b_0(t)$ and $b_1(t)$ are token variables for two constraints, respectively, taking value in the set $B=\{0,1,2,\ldots,b_{max}\}$. $\Delta(t) \in \{1,2,3,\cdots,\Delta_{max}\}$ is the AoI at the receiver, and $r(t) \in \{0,1\}$ is the request process at time slot $t$; $r(t)$ is $1$ when there is a request from the destination node and $0$ otherwise. The state space, $\hat{\mathcal{S}}\!=\!\big\{(b_0,b_1,\Delta,r)\!:\! \ b_0 \in B, b_1 \in B,\Delta \in \{1,2,\cdots\!,\Delta_{max}\} \text{, and } r \in \{0,1\} \big\}$ is a finite set. The time evolution of these state variables is given by the following equations:
		
		{\footnotesize
			\begin{align}
				\label{b0Evol_eqn}
				b_0(t\!+\!1)\!=\!
				\begin{cases}
					\min\left\{b_0(t) \!-\! a(t) \!+\! 1,b_{\text{max}}\right\} & \parbox[t][][t]{2.5 cm}{\footnotesize {with probability (w.p.) $\alpha_{\text{min}}$ when $r(t)=0$, }}   \\
					b_0(t) - a(t) & \parbox[t][][t]{2.5 cm}{\footnotesize {w.p. $1-\alpha_{\text{min}}$ when $r(t)=0$, }} \\
					b_0(t) & \parbox[t][][t]{2.5 cm}{\footnotesize {w.p. $1$ when $r(t)=1$.}}
				\end{cases}
			\end{align}

			\begin{align}
				\label{b1Evol_eqn}
				b_1(t\!+\!1)\!=\!
				\begin{cases}
					\min\left\{b_1(t) \!-\! a(t) \!+\! 1,b_{\text{max}}\right\} & \parbox[t][][t]{2.5 cm}{\footnotesize {w.p. $\alpha_{\text{max}}$ when $r(t)=1$, }}   \\
					b_1(t) - a(t) & \parbox[t][][t]{2.5 cm}{\footnotesize {w.p. $1-\alpha_{\text{max}}$ when $r(t)=1$, }} \\
					b_1(t) & \parbox[t][][t]{2.5 cm}{\footnotesize {w.p. $1$ when $r(t)=0$.}}
				\end{cases}
			\end{align}

			\begin{align}
				\label{AoI_Evol_eqn}
				\Delta(t\!+\!1)\!=\!
				\begin{cases}
					\min\left\{\Delta(t)\!+\!1,\Delta_{max}\right\} & \text{when } a(t)\!=\!0, \\
					1 & \text{when } a(t)\!=\!1.
				\end{cases}
			\end{align}

			\begin{align}
				\label{r_Evol_eqn}
				r(t+1)\!=\!
				\begin{cases}
					1 & \text{w.p. } q, \\
					0 & \text{w.p. } 1-q.
				\end{cases}
			\end{align}
		}
		
		\item Actions: at time $t$, $a^\pi(t)=0$ represents the action of staying idle, while $a^\pi(t)=1$ represents the action of transmitting an update. The action $a(t)$ is forced to be $0$ when there is no token, i.e., when $b_0(t)=0$ and $r(t)=0$, or $b_1(t)=0$ and $r(t)=1$.
		\item Transition probabilities: given the following equation, 
		\begin{align}
			P&\left[\hat{s}(t+1)| \hat{s}(t),a(t)\right]\!=\!P\left[b_0(t+1)|b_0(t),r(t),a(t)\right] \notag \\
			&\times P\left[b_1(t+1)|b_1(t),r(t),a(t)\right] \notag \\ 
			&\times  P\left[\Delta(t+1)|\Delta(t),a(t)\right]\times P\left[r(t+1)\right],
		\end{align}
		the transition probabilities can simply be written using the equations \eqref{b0Evol_eqn} to \eqref{r_Evol_eqn}.
		\item Cost function: the cost function is identical to the primary CMDP problem: $C\big( \hat{s}(t),a(t)\big)\overset{\text{def}}{=}\Delta(t)$.
	\end{itemize}
	
	The resulting token-based unconstrained MDP optimization problem is given by,
	\begin{align}
		\label{UncMDP_eqn_2rate}
		\min_{\pi \in \Pi}\ \limsup_{T\rightarrow\infty} {\frac{1}{T} E\Bigg[ \sum_{t=0}^{T-1} C\Big( \hat{s}^\pi(t),a^\pi(t) \Big) \Big| \hat{s}(0) \Bigg]},
	\end{align}
	where $\Pi$ is the set of all feasible policies, and $\hat{s}(0)$ is the initial state of the system. In what follows, we first present analytical results for this problem. In Section \ref{TwoRateNumericalSection}, we will compare the optimal solution with the optimal CMDP solution obtained through a proposed iterative algorithm.
	
	\subsection{Analytical Results} 
	
	In this section, we provide analytical results concerning the existence and structure of the optimal policy for the token-based MDP problem \eqref{UncMDP_eqn_2rate}. Here, we omit the \emph{hat} on $\hat{s}$ for the sake of simplicity.
	
	\begin{proposition}
		The token-based MDP problem \eqref{UncMDP_eqn_2rate} is weakly accessible.
	\end{proposition}
	
	\begin{proof}
		We establish that any state $s^\prime=\left(b_0^\prime,b_1^\prime,r^\prime,\Delta^\prime\right) \in \hat{\mathcal{S}}$ is reachable from any other state $s=\left(b_0,b_1,r,\Delta\right) \in \hat{\mathcal{S}}$ under a stationary stochastic policy $\pi$, where the action $a\in\{0,1\}$ at each state is randomly selected.
		The token state $b_0^\prime<b_0$ ($b_1^\prime<b_1$) can be accessed from $b_0$ ($b_1$) with a positive probability by executing action $a=1$ for $b_0-b_0^\prime$ ($b_1-b_1^\prime$) time slots. Conversely, $b_0^\prime \geq b_0$ ($b_1^\prime \geq b_1$) is reachable from $b_0$ ($b_1$) with a positive probability (w.p.p.) by implementing action $a=0$ for $b_0^\prime-b_0$ ($b_1^\prime-b_1$) slots. Upon reaching the state $b_0^\prime$ ($b_1^\prime$), the subsequent actions have no impact, and the token variables' state can remain unchanged with positive probability. Consequently, for the remainder of the proof, we consider the tokens' state as $(b_0^\prime,b_1^\prime)$. Additionally, the state $r^\prime \in \{0,1\}$ is reached w.p.p. irrespective of the previous or current system state.
		Finally, the state $\Delta^\prime<\Delta$ is attainable from $\Delta$ w.p.p. by executing action $a=1$ for one time slot, followed by executing action $a=0$ for $\Delta^\prime$ slots. Meanwhile, the state $\Delta^\prime \geq \Delta$ is accessible from $\Delta$ by executing action $a=0$ for $\Delta^\prime-\Delta$ slots. Consequently, $s^\prime$ is accessible from $s$, thereby concluding the proof.
	\end{proof}
	
	\begin{proposition}
		In the token-based MDP problem \eqref{UncMDP_eqn_2rate}, the optimal average cost $J^\ast$ achieved by an optimal policy $\pi^\ast$ is the same for all initial states, satisfying the Bellman’s equations given by \eqref{Bellman_eqn} and \eqref{OptimalAction_eqn}.
	\end{proposition}
	
	\begin{proof}
		Given that the MDP problem \eqref{UncMDP_eqn_2rate} is weakly accessible, the proof is the same as the proof of Proposition 2, which has been omitted for the sake of space. 
	\end{proof}
	
	Standard methods, like the (Relative) Value Iteration and Policy Iteration algorithms, can be employed to solve this optimization problem.
	
	\begin{theorem}
		The optimal policy $\pi^\ast$ of the token-based MDP problem \eqref{UncMDP_eqn_2rate} is a threshold policy. This means that for each combination of $(b_0,b_1,r)$ there exists an age threshold $\Delta_{T}$ such that $\pi^\ast(b_0,b_1,\Delta,r)$ is $1$ for $\Delta \geq \Delta_{T}$, and $0$ otherwise.
	\end{theorem}
	
	\begin{proof}
		The proof is the same as the proof of Theorem 1, which has been omitted for the sake of space. 
	\end{proof}
	
	\subsection{Iterative triangle bisection algorithm for solving the CMDP problem}
	\label{IteratibeTriangleSection}
	We present an algorithm to obtain the optimal solution for the problem \eqref{CMDP_eqn_2rate}. 
	For a CMDP with a single constraint, the optimal policy can be found by applying an iterative algorithm approach as presented in \cite{hatami2022demand}. 
	The algorithm optimizes the Lagrange dual problem \eqref{eqn:DualMDP} resulting from the CMDP by iterating through two loops: the inner loop computes the $\bm{\lambda}$-optimal policy for a given Lagrange multiplier $\bm{\lambda}$ using the Relative Value Iteration Algorithm (RVIA), and the external loop finds the optimal Lagrange multiplier $\bm{\lambda}^\ast$ for a given update policy through a bisection search. However, the aforementioned bisection search is not applicable to cases with two Lagrangian multipliers. 
	To address this issue in the external loop, we provide a generalized version of the bisection method referred to as \emph{iterative triangle bisection}, which can be employed for solving two-dimensional nonlinear equations. This method is a concise version of the algorithm introduced in \cite{harvey1976two} and outlines an algorithm for solving a system $\bm{F}_{\bm{\lambda}}\overset{\text{def}}{=}$$\begin{bsmallmatrix} c_0(\pi_{\bm{\lambda}}^\ast) \\ c_1(\pi_{\bm{\lambda}}^\ast) \end{bsmallmatrix}$$=$$\begin{bsmallmatrix} 0 \\ 0 \end{bsmallmatrix}$$ $. 
	In this method, a triangle ($\mathit{\Delta}\bm{\lambda}_A\bm{\lambda}_B\bm{\lambda}_C$) in a two-dimensional plane is iteratively bisected into two triangles ($\mathit{\Delta}\bm{\lambda}_A\bm{\lambda}_D\bm{\lambda}_C$ and $\mathit{\Delta}\bm{\lambda}_D\bm{\lambda}_B\bm{\lambda}_C$), and the search involves identifying the triangle containing a point $\bm{\lambda}_E \overset{\text{def}}{=} (\lambda_{0E},\lambda_{1E})$ that satisfies $\bm{F}_{\bm{\lambda}_E} = \begin{bsmallmatrix} c_0(\pi_{\bm{\lambda}_E}^\ast) \\ c_1(\pi_{\bm{\lambda}_E}^\ast) \end{bsmallmatrix}$$=$$\begin{bsmallmatrix} 0 \\ 0 \end{bsmallmatrix}$$ $. 
	The verification that the point $\begin{bsmallmatrix} 0 \\ 0 \end{bsmallmatrix}$ lies within a triangle $\mathit{\Delta}\bm{F}_{\bm{\lambda}_{A}}\bm{F}_{\bm{\lambda}_{D}}\bm{F}_{\bm{\lambda}_{C}}$ can be accomplished using various methods, such as the \emph{L-test} introduced in \cite{harvey1976two}. 
	The proposed method is presented in Algorithm \ref{alg_CMDP}. We have also provided the inner loop of RVIA in Algorithm \ref{alg_RVIA}. 
	Since the action space is discrete, the optimal policy $\pi_{\bm{\lambda}}^\ast$ may not yield the maximum allowable rates but rather achieves approximately equal rates. For justification, and in the case of a CMDP with a single constraint, it is well established that the ultimate optimal policy is a mixture of two non-randomized stationary policies \cite{beutler1985optimal}: one policy is tailored to values exceeding (albeit nearly equal to) the constraint, while the other is designed for values below (yet again, nearly equal to) the maximum allowable rate. Through this approach, upon mixing, the equality constraint can be met. 
	Here, we adopt a generalized approach where a set of four \emph{mixed policies} \cite[Section 6.3]{altman1999constrained}, denoted by $\pi_{\bm{\lambda}^{\!+\!+}}^\ast$, $\pi_{\bm{\lambda}^{\!+\!-}}^\ast$, $\pi_{\bm{\lambda}^{\!-\!+}}^\ast$ and $\pi_{\bm{\lambda}^{\!-\!-}}^\ast$, is employed. These four policies represent the four nearest $\bm{\lambda}$-optimal policies, achieved through a gradual decrease (increase) of $\lambda_{0}^\ast$ and/or $\lambda_{1}^\ast$ until the conditions of the greater than (smaller than) inequalities are met, such that $\begin{bsmallmatrix} c_0(\pi_{\bm{\lambda}^{\!+\!+}}^\ast) \geq 0 \\ c_1(\pi_{\bm{\lambda}^{\!+\!+}}^\ast) \geq 0 \end{bsmallmatrix}$,  $\begin{bsmallmatrix} c_0(\pi_{\bm{\lambda}^{\!-\!-}}^\ast) \leq 0 \\ c_1(\pi_{\bm{\lambda}^{\!-\!-}}^\ast) \leq 0 \end{bsmallmatrix}$, $\begin{bsmallmatrix} c_0(\pi_{\bm{\lambda}^{\!+\!-}}^\ast) \geq 0 \\ c_1(\pi_{\bm{\lambda}^{\!+\!-}}^\ast) \leq 0 \end{bsmallmatrix}$, and $\begin{bsmallmatrix} c_0(\pi_{\bm{\lambda}^{\!-\!+}}^\ast) \leq 0 \\ c_1(\pi_{\bm{\lambda}^{\!-\!+}}^\ast) \geq 0 \end{bsmallmatrix}$$ $. 
	The gradual decrease (increase) of the $\lambda$ multipliers can simply be achieved by iteratively dividing (multiplying) them by $1 + \gamma$, where $\gamma$ is chosen to be a small positive real number close to zero, e.g., $0.1$. The mixing probability of these policies can be expressed by introducing probabilities $\rho_0$ and $\rho_1$ where $\pi_{\bm{\lambda}^{\!+\!+}}^\ast$, $\pi_{\bm{\lambda}^{\!+\!-}}^\ast$, $\pi_{\bm{\lambda}^{\!-\!+}}^\ast$ and $\pi_{\bm{\lambda}^{\!-\!-}}^\ast$ are assigned probabilities  $\rho_0\rho_1$, $\rho_0(1-\rho_1)$, $(1-\rho_0)\rho_1$, and $(1-\rho_0)(1-\rho_1)$, respectively, satisfying the following system of equations:
	\begin{align}
		\rho_0\rho_1c_0(\pi_{\bm{\lambda}^{\!+\!+}}^\ast) &\!+\! \rho_0(1\!-\!\rho_1)c_0(\pi_{\bm{\lambda}^{\!+\!-}}^\ast) \!+\! (1\!-\!\rho_0)\rho_1c_0(\pi_{\bm{\lambda}^{\!-\!+}}^\ast) \notag \\
		& \!+\! (1\!-\!\rho_0)(1\!-\!\rho_1)c_0(\pi_{\bm{\lambda}^{\!-\!-}}^\ast) \!=\! 0 \\
		\rho_0\rho_1c_1(\pi_{\bm{\lambda}^{\!+\!+}}^\ast) &\!+\! \rho_0(1\!-\!\rho_1)c_1(\pi_{\bm{\lambda}^{\!+\!-}}^\ast) \!+\! (1\!-\!\rho_0)\rho_1c_1(\pi_{\bm{\lambda}^{\!-\!+}}^\ast) \notag \\
		& \!+\! (1\!-\!\rho_0)(1\!-\!\rho_1)c_1(\pi_{\bm{\lambda}^{\!-\!-}}^\ast) \!=\! 0
	\end{align}
	
	\subsection{Complexity of RVIA and iterative triangle bisection}
	\label{ComplexitySection}
	The RVIA has a computational complexity of $\mathcal{O}\left(N_{\text{RVIA}} |\mathcal{S}|^2 |\mathcal{A}| \right)$, where $|\mathcal{S}|$ is the size of the state space, $|\mathcal{A}|$ is the size of the action space, and $N_{\text{RVIA}}$ is the number of iterations for termination.
	The computational complexity of the iterative triangle bisection algorithm is $\mathcal{O}\left(N_{\text{E}} N_{\text{RVIA}} |\mathcal{S}|^2 |\mathcal{A}| \right)$, where $N_{\text{E}}$ is the iteration number of the external loop. Both $N_{\text{E}}$ and $N_{\text{RVIA}}$ depend on initial points and termination parameters $\varepsilon_{\lambda}$ and $\varepsilon_V$, respectively.
	For the CMDP problem with two rate constraints \eqref{CMDP_eqn_2rate}, the iterative triangle bisection algorithm's complexity is $\mathcal{O}\left(N_{\text{E}} N_{\text{RVIA}} |\Delta_{\text{max}}|^2 \right)$. For the corresponding token-based MDP, the RVIA complexity is $\mathcal{O}\left(\hat{N}_{\text{RVIA}} |b_{\text{max}}|^4 |\Delta_{\text{max}}|^2 \right)$. 
	
	\begin{algorithm}[!t]
		\caption{Iterative triangle bisection approach to solve CMDP problem with two constraints}
		\small
		\label{alg_CMDP}
		\begin{algorithmic}[1]
			\Require {System parameters $q$, $\alpha_{\text{min}}$, $\alpha_{\text{max}}$, $\Delta_{\text{max}}$, $\varepsilon_{\lambda}$ and a set of four initial points $\bm{\lambda}_A$, $\bm{\lambda}_B$, $\bm{\lambda}_C$, and $\bm{\lambda}_D$ such that $\begin{bsmallmatrix} c_0(\pi_{\bm{\lambda}_A}^\ast) \geq 0 \\ c_1(\pi_{{\bm{\lambda}}_A}^\ast) \geq 0 \end{bsmallmatrix}$,  $\begin{bsmallmatrix} c_0(\pi_{\bm{\lambda}_B}^\ast) \leq 0 \\ c_1(\pi_{\bm{\lambda}_B}^\ast) \leq 0 \end{bsmallmatrix}$, $\begin{bsmallmatrix} c_0(\pi_{\bm{\lambda}_C}^\ast) \geq 0 \\ c_1(\pi_{\bm{\lambda}_C}^\ast) \leq 0 \end{bsmallmatrix}$, and $\begin{bsmallmatrix} c_0(\pi_{\bm{\lambda}_D}^\ast) \leq 0 \\ c_1(\pi_{\bm{\lambda}_D}^\ast) \geq 0 \end{bsmallmatrix}$.}
			\State {Initialize the triangles: $R \leftarrow \mathit{\Delta}\bm{\lambda}_A\bm{\lambda}_D\bm{\lambda}_C$, $S \leftarrow \mathit{\Delta}\bm{\lambda}_D\bm{\lambda}_B\bm{\lambda}_C$. 
				\State Initialize the $\bm{\lambda}$ vector: $\bm{\lambda}_E^{\text{new}} \leftarrow \bm{\lambda}_D$, $\bm{\lambda}_E^{\text{old}} \leftarrow \begin{bsmallmatrix} \inf \\ \inf \end{bsmallmatrix}$.}
			
			\While{$|\bm{\lambda}_E^{\text{new}}\!-\bm{\lambda}_E^{\text{old}}| \!\geq\! \varepsilon_{\lambda}$} \Comment{{\footnotesize \textit{External loop: Triangle bisection}}}
			\State {Run $\texttt{RVIA}(\bm{\lambda})$ at the vertices of $R$ ($\bm{\lambda}_{RA}$, $\bm{\lambda}_{RB}$, and $\bm{\lambda}_{RC}$) to obtain $\bm{\lambda}$-optimal policies: $\pi_{\bm{\lambda}_{RA}}^\ast \leftarrow \texttt{RVIA}({\bm{\lambda}_{RA}})$, $\pi_{\bm{\lambda}_{RB}}^\ast \leftarrow \texttt{RVIA}({\bm{\lambda}_{RB}})$, and $\pi_{\bm{\lambda}_{RC}}^\ast \leftarrow \texttt{RVIA}({\bm{\lambda}_{RC}})$.} 
			\State Calculate $\bm{F}_{\bm{\lambda}}=\begin{bsmallmatrix} c_0(\pi_{\bm{\lambda}}^\ast) \\ c_1(\pi_{\bm{\lambda}}^\ast) \end{bsmallmatrix}$ at the vertices of $R$ to obtain triangle $\mathit{\Delta}\bm{F}_{\bm{\lambda}_{RA}}\bm{F}_{\bm{\lambda}_{RB}}\bm{F}_{\bm{\lambda}_{RC}}$ 
			\If {$\left(\bm{0} \text{  lies within triangle  } \mathit{\Delta}\bm{F}_{\bm{\lambda}_{RA}}\bm{F}_{\bm{\lambda}_{RB}}\bm{F}_{\bm{\lambda}_{RC}}\right)$} 
			$\ T \leftarrow R$ \Else $\ T \leftarrow S$ \EndIf
			\State Rotate $T$ such that its first edge ($\bm{\lambda}_{TA}\bm{\lambda}_{TB}$) becomes the longest edge of triangle $T$.
			\State Update $\mathit{\Delta}\bm{\lambda}_A\bm{\lambda}_B\bm{\lambda}_C \leftarrow T$
			\State Update $\bm{\lambda}_E^{\text{old}} \!\leftarrow\! \bm{\lambda}_E^{\text{new}}$, $\bm{\lambda}_E^{\text{new}} \!\leftarrow\! \frac{\bm{\lambda}_A+\bm{\lambda}_B+\bm{\lambda}_C}{3}$, and $\bm{\lambda}_D \!\leftarrow\! \frac{\bm{\lambda}_A+\bm{\lambda}_B}{2}$
			\State Update the triangles: $R \leftarrow \mathit{\Delta}\bm{\lambda}_A\bm{\lambda}_D\bm{\lambda}_C$, $S \leftarrow \mathit{\Delta}\bm{\lambda}_D\bm{\lambda}_B\bm{\lambda}_C$.
			\EndWhile
			\State $\bm{\lambda}^\ast \leftarrow \bm{\lambda}_E^{\text{new}}$ and $\pi_{\bm{\lambda}}^\ast \leftarrow \texttt{RVIA}({\bm{\lambda}}^\ast$) \\
			\Return $\bm{\lambda}^\ast$ and $\pi_{\bm{\lambda}}^\ast$
		\end{algorithmic}  
	\end{algorithm}
	
	\begin{algorithm}[!t]
		\caption{\texttt{RVIA}(${\bm{\lambda}}$) function}
		\label{alg_RVIA}
		{\small
			\begin{algorithmic}[1]
				\Require State space $S$, parameters $q$, $\varepsilon_V$, and the input $\bm{\lambda}$
				\State Initialize $ V_0(s)=v_0(s)=0,\ \forall s \in S $, and set $t=0$.
				\Repeat \Comment{{\footnotesize \textit{Inner loop: RVIA iteration}}}
				\State $ t\leftarrow{t+1} $
				\For{$ s \in S $}
				\State $ {\displaystyle v_t(s)=\min_{a\in\left\{0,1\right\}}{\sum_{s^\prime\in S}{P\left(s^\prime|s,a\right)\left[C_{\bm{\lambda}}(s,a)+V_{t-1}(s^\prime)\right]}} }$ 
				\State $ {\displaystyle \pi_{\bm{\lambda}}^\ast(s)=\argmin_{a\in\left\{0,1\right\}}{\sum_{s^\prime\in S}{P\left(s^\prime|s,a\right)\left[C_{\bm{\lambda}}(s,a)+V_{t-1}(s^\prime)\right]}} }$ 
				\State $V_t(s)=v_t(s)-v_t(s_0)$ \Comment {{\footnotesize\textit{$s_0$: any arbitrary state}}}
				\EndFor
				\Until{$ {\displaystyle \max_{s\in S}\left\{V_t(s)\!-\!V_{t-1}(s)\right\}\!-\!\min_{s\in S}\left\{V_t(s)\!-\!V_{t-1}(s)\right\} \!<\! \varepsilon_V } $}  \\
				\Return $ \pi_{\bm{\lambda}}^\ast$
			\end{algorithmic}
		}
	\end{algorithm}

	\section{Numerical Results}
	\label{NumericalResults}
	In this section, we evaluate the performance of the proposed token-based policy through simulation results. We execute our algorithms on MATLAB over $2\times10^4$ time slots and average them over $400$ runs. All numerical results are obtained using a standard laptop with an Intel(R) Core(TM) i7-1355U 1.7 GHz processor and 32 GB of RAM.
	
	\subsection{Single Rate Constrained Problem}
	\label{SingleRateNumericalSection}
	We compare the results of the optimal policies for the primary CMDP problem \eqref{CMDP_eqn}, and the corresponding token-based MDP problem \eqref{UncMDP_eqn}. The first policy is derived through the \emph{optimal threshold finder} presented in \cite{maatouk2020age}, while the second policy is obtained using the RVIA algorithm. In Figs. \eqref{fig:AvgVsAlpha} and \eqref{fig:AvgVsPr} the average AoII for both policies is depicted for various sets of system parameters. The optimal average AoII for various values of $\alpha$ has been illustrated in Fig. \ref{fig:AvgVsAlpha}, with $p_R=0.5$ and $N=8$. In this figure, we have modified the value of $b_{\text{max}}$ from $5$ to $20$. As evident, the optimal AoII for the token-based MDP problem converges to the optimal AoII of the main CMDP problem.
	In Fig. \ref{fig:AvgVsPr}, the optimal average AoII for two policies is depicted as a function of $p_R$, for $N=8$ and $\alpha=0.1$. The optimal average AoII for the token-based policy converges to the optimal average AoII of the main CMDP problem as the value of $b_{\text{max}}$ increases. In both figures, we considered $\Delta_{\text{max}} = 30$.
	\begin{figure}[hbt!]
		\centering
		\includegraphics[width=2.3in,trim={0cm 0cm 0cm 1.1cm}]{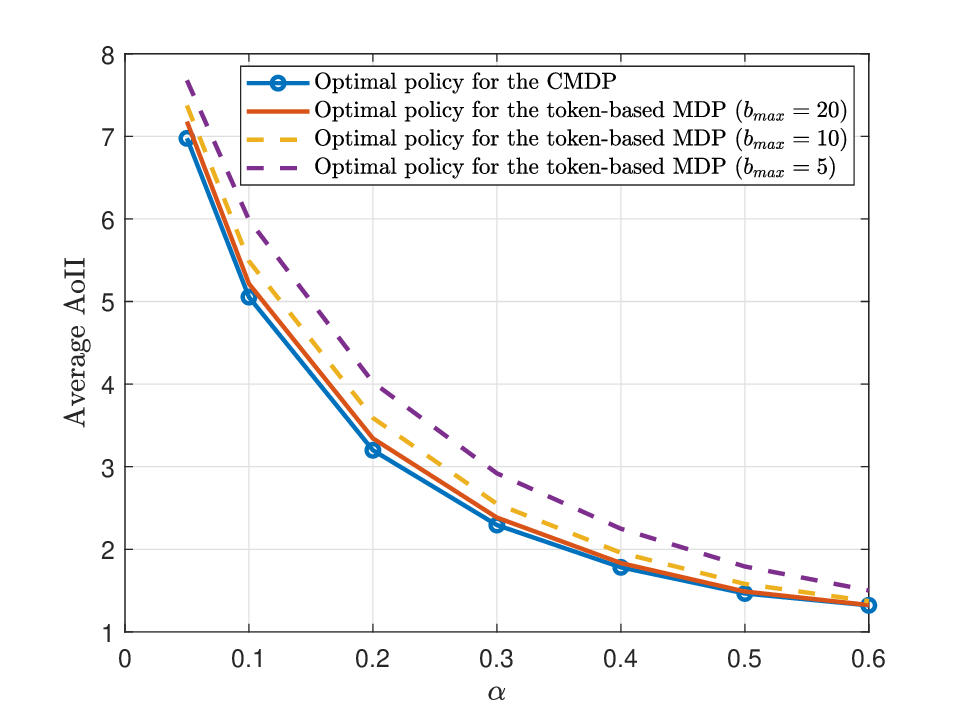}
		\caption{Optimal average AoII for two policies vs. $\alpha$ in \\the single rate constrained problem.}
		\label{fig:AvgVsAlpha}
		\vspace{1pt}
	\end{figure}
	\begin{figure}[hbt!]
		\centering
		\includegraphics[width=2.3in,trim={0cm 0cm 0cm 0.6cm}]{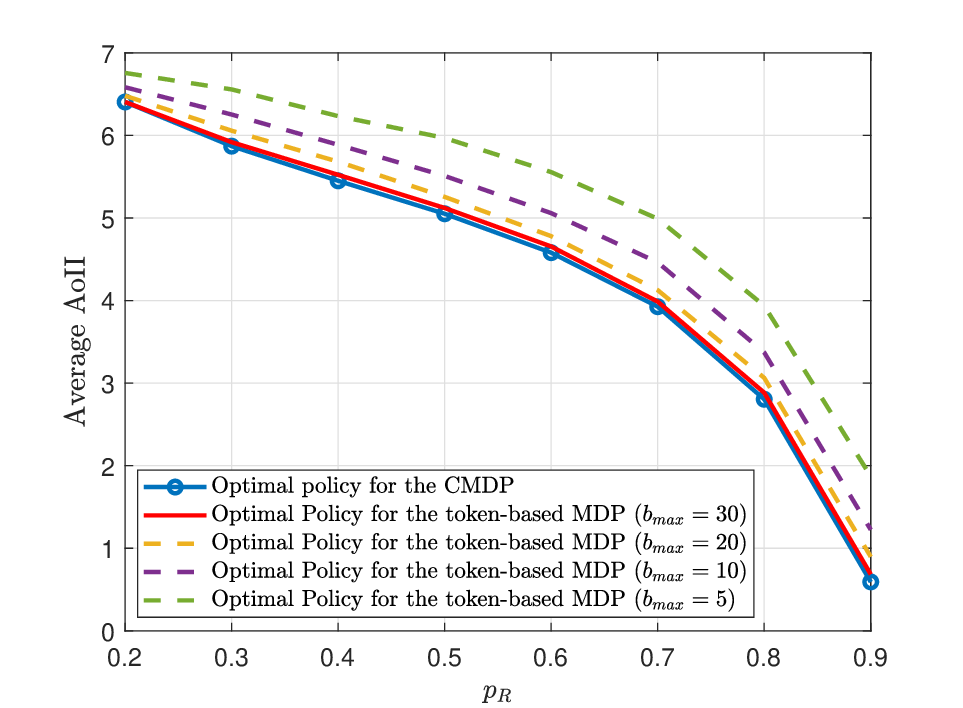}
		\caption{Optimal average AoII for two policies vs. $p_R$ in \\the single rate constrained problem.}
		\label{fig:AvgVsPr}
		\vspace{8pt}
	\end{figure}
	
	\subsection{Two Rate Constrained Problem}
	\label{TwoRateNumericalSection}
	
	We compare the optimal token-based policy of the MDP problem \eqref{UncMDP_eqn_2rate} with the optimal policy derived from the primary CMDP problem \eqref{CMDP_eqn_2rate}, along with two additional baseline policies.
	We utilize the RVIA algorithm to solve the token-based MDP problem and the iterative triangle bisection algorithm to solve the CMDP problem. Subsequently, we compare the resulting average AoI for these two policies across various system parameters $q$ (Fig. \ref{fig:AvgAoIVsq}) and $\alpha_{\text{max}}$ (Fig. \ref{fig:AvgAoIVsAlpha}), while keeping $\alpha_{\text{min}}$ fixed at $0.1$. In Figs. \ref{fig:AvgAoIVsq} and \ref{fig:AvgAoIVsAlpha}, we also present the average AoI resulting from a uniform two-rate policy and a random policy. By a uniform two-rate policy, we mean that when $r=0$, the system is updated at a rate of $\alpha_{\text{min}}$, and when $r=1$, the system is updated at a rate of $\alpha_{\text{max}}$. Whereas, according to a random policy, the system is updated with probability $\alpha_{\text{min}}$ ($\alpha_{\text{max}}$) when $r=0$ ($r=1$). $\Delta_{\text{max}}$, $\varepsilon_{\bm{\lambda}}$, and $\varepsilon_V$ have been set to $20$, $0.1$, and $0.1$, respectively. The results are summarized as follows:
	
	\begin{itemize}[leftmargin=0.12in]
		\item With an increase in $b_{\text{max}}$, the optimal average AoI for the token-based MDP problem converges to that of the primary CMDP problem. For further illustration, in Fig. \ref{fig:GapAoIVsbmax}, we plot the \emph{optimality gap} between the token-based policy and the optimal policy for the main CMDP problem as a function of $b_{\text{max}}$, where $\alpha_{\text{max}}=0.5$ and $q$ is either $0.2$ or $0.5$.
		\item Even for low levels of $b_{\text{max}}$ (e.g. $b_{\text{max}}=5$), the performance of the optimal token-based policy is close to the optimal CMDP policy, and it outperforms both the uniform and random policies. In Section \ref{ComplexitySection}, we observed that when $b_{\text{max}}$ is low, the RVIA's complexity in finding the optimal token-based policy may be preferable to the iterative approach for the primary CMDP. This is particularly significant because \emph{the token-based approach offers a straightforward formulation and solution without becoming entangled in the Lagrangian formulation and the intricate aspects of the iterative bisection algorithm for the primary CMDP problem.}
	\end{itemize}
	\vspace{-4pt}
	\begin{figure}[hbt!]
		\centering
		\includegraphics[width=2.3in,trim={0cm 0cm 0cm 1.3cm}]{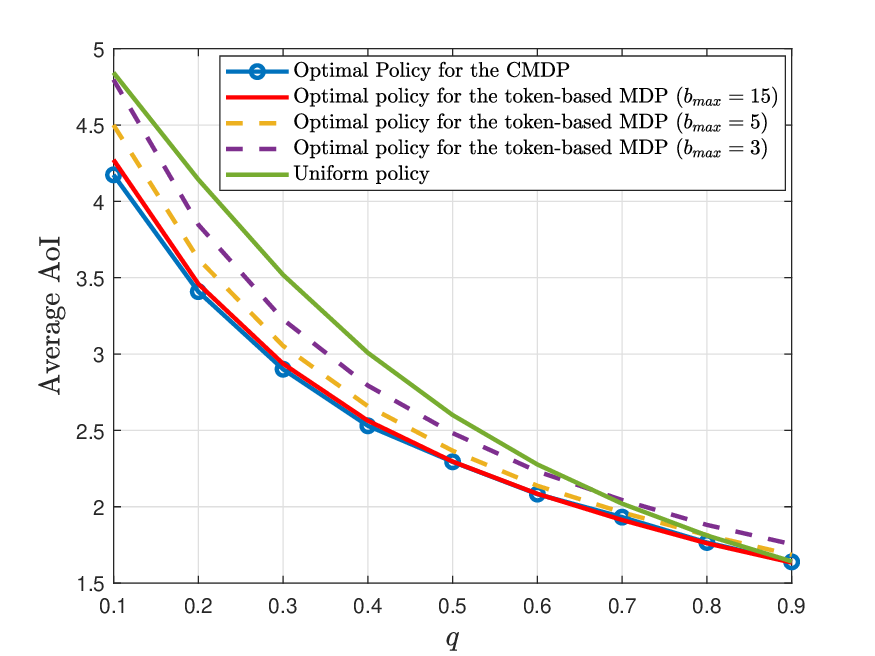}
		\caption{Average AoI for different policies vs. $q$ ($\alpha_{\text{max}}=0.5$) in \\the two rate constrained problem.}
		\label{fig:AvgAoIVsq}
	\end{figure}
	
	\begin{figure}[hbt!]
		\centering
		\includegraphics[width=2.3in,trim={0cm 0cm 0cm 1.2cm}]{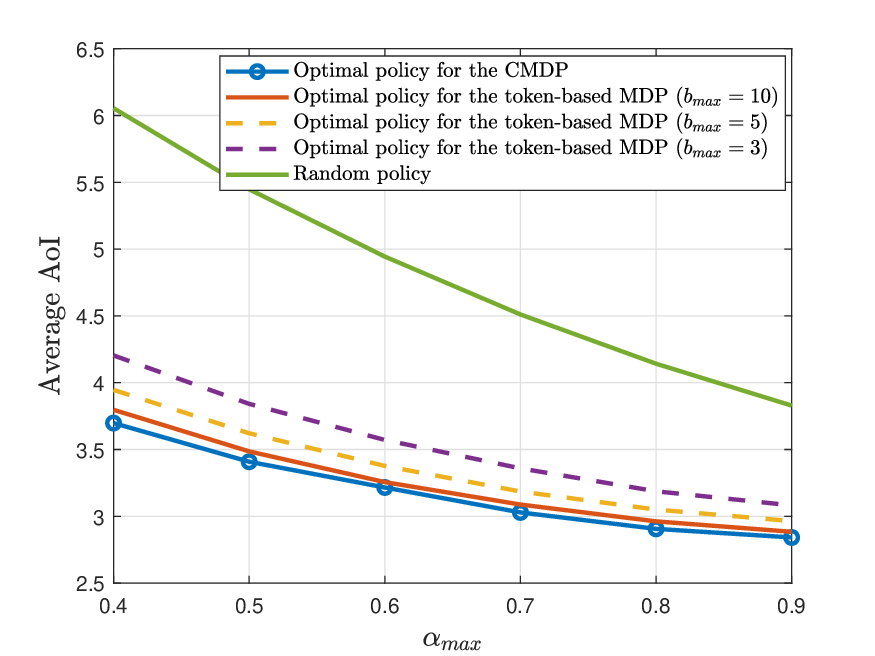}
		\caption{Average AoI for different policies vs. $\alpha_{\text{max}}$ ($q=0.2$) in \\the two rate constrained problem.}
		\label{fig:AvgAoIVsAlpha}
	\end{figure}
	
	\begin{figure}[hbt!]
		\centering
		\includegraphics[width=2.3in,trim={0cm 0cm 0cm 1.2cm}]{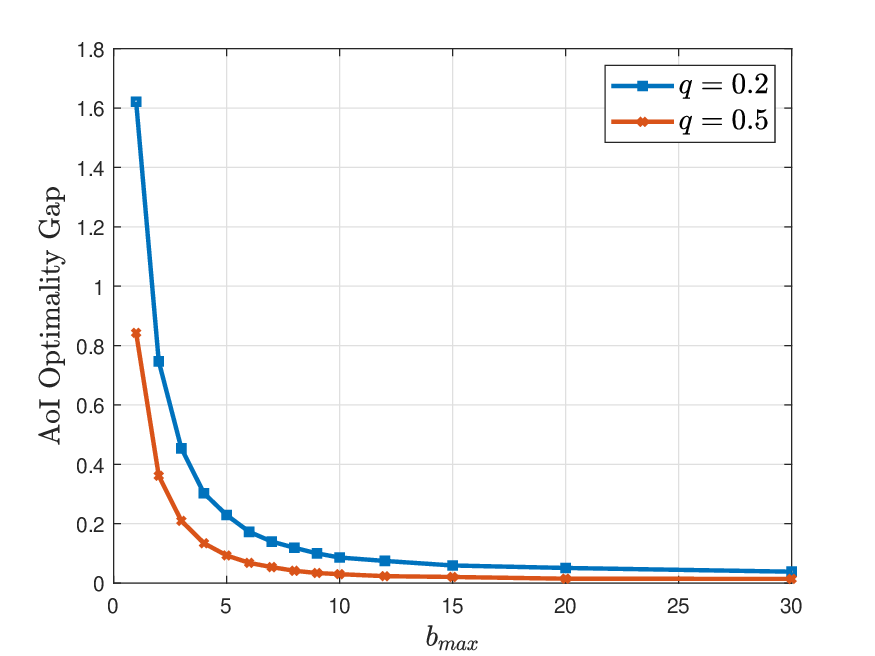}
		\caption{Optimality gap vs. $b_{\text{max}}$ in the two rate constrained problem.}
		\label{fig:GapAoIVsbmax}
		\vspace{-2pt}
	\end{figure}

	\section{Conclusion}
	\label{ConclusionSection}
	In this work, we explored a token-based approach for transforming CMDP problems into unconstrained MDP problems, explicitly focusing on optimizing information freshness in status update systems. We have applied this approach to systems with one and two update rate constraints and examined the analytical structure of the optimal token-based policy. Additionally, we have introduced an iterative triangle bisection algorithm for solving CMDP problems with two constraints. Our findings demonstrate the superior performance of the optimal policy compared to baseline policies and the convergence of the optimal token-based policy to the optimal policy of the main CMDP problem as the maximum number of tokens in the system increases.
	
	\bibliographystyle{IEEEtran}
	\bibliography{bibliography_short}

	\appendices
	
	\section{Proof of Theorem 1}
	\label{Apen1:Theorem1}
	\begin{proof}
		The Bellman equation at state $s\!=\!\left(b,\Delta\right)$ can be simplified as follows:
		\begin{align}
			J^\ast+V(s)&=\Delta+\min_{a\in\left\{0,1\right\}}{\left\{\sum_{s^\prime\in S} P\left[s^\prime|s,a\right]V(s^\prime)\right\}}.
		\end{align}
		
		Therefore, the optimal action can be obtained by:
		\begin{align*}
			a^\ast(s)\!=\!\argmin_{a\in\left\{0,1\right\}}{\left\{\sum_{s^\prime\in S}\!P\left[s^\prime|s,a\right]V(s^\prime)\right\}}\!=\!
			\begin{cases}
				0, & DV(s)\!\geq \!0,\\
				1, & DV(s)\!<\!0,\\
			\end{cases}
		\end{align*}
		where $V^0(s)\overset{\text{def}}{\!=\!}\sum_{s^\prime\in S}P\left[s^\prime|s,a\!=\!0\right]V(s^\prime)$, $V^1(s)\overset{\text{def}}{\!=\!}\sum_{s^\prime\in S}\!P\left[s^\prime|s,a\!=\!1\right]V(s^\prime)$, and $DV(s)\!\overset{\text{def}}{=}\!V^1(s)\!-\!V^0(s)$.
		
		As can be seen, the optimal action $a^\ast(s)$ is related to the sign of $DV(s)$. When $b=0$ or $\Delta=0$, it can be simply shown that $DV(s) = 0$; thus the action $a=0$ is optimal. For other cases where $b>0$ and $\Delta>0$, we have:
		\begin{subequations}
			\label{eqn:DV_1rate}
			\begin{align}
				V^0(s)& \!=\! \alpha p_t V(b\!+\!1,0) \!+\! \alpha \left(1\!-\!p_t\right) V(b\!+\!1,\Delta\!+\!1) \\ 
				& + (1-\alpha) p_t V(b,0) + (1-\alpha) \left(1-p_t\right) V(b,\Delta+1), \notag \\
				V^1(s)& = \alpha \beta V(b,0) + \alpha  \left(1-\beta\right) V(b,\Delta+1)  \\
				& + (1\!-\!\alpha) \beta V(b\!-\!1,0) + (1\!-\!\alpha) \left(1\!-\!\beta\right) V(b\!-\!1,\Delta\!+\!1), \notag
			\end{align}
		\end{subequations}
		
		In what follows, we demonstrate that $DV(s)=DV(b,\Delta)= V^1(s) - V^0(s)$ is a decreasing function of $\Delta$, i.e., for $\Delta\!^- \leq \Delta\!^+$, we show that $DV(b,\Delta\!^+) \leq DV(b,\Delta\!^-)$ or $DV(b,\Delta\!^+) - DV(b,\Delta\!^-) \leq 0$. By simplification of $DV(b,\Delta\!^+)$ and $DV(b,\Delta\!^-)$ based on \eqref{eqn:DV_1rate} we obtain: 
		\begin{align}
			\label{eqn:DV_diff}
			&DV(b,\Delta\!^+) - DV(b,\Delta\!^-) \notag \\
			& = \alpha \Big\{ (1\!-\!\beta) \big[V(b,\Delta\!^+ \!+\!1) \!-\! V(b,\Delta\!^- \!+\!1) \big] \notag \\
			& \quad -\!  (1\!-\!p_t) \big[ V(b\!+\!1,\Delta\!^+ \!+\!1) \!-\! V(b\!+\!1,\Delta\!^- \!+\!1) \big] \Big\} \notag \\
			& \quad + (1\!-\!\alpha) \Big\{ (1\!-\!\beta) \big[V(b\!-\!1,\Delta\!^+ \!+\!1) \!-\! V(b\!-\!1,\Delta\!^- \!+\!1) \big] \notag \\
			& \quad -\! (1\!-\!p_t) \big[ V(b,\Delta\!^+ \!+\!1) \!-\! V(b,\Delta\!^- \!+\!1) \big] \Big\}.
		\end{align}
		
		According to \eqref{eqn:DV_diff}, to confirm the inequality $DV(b,\Delta\!^+) - DV(b,\Delta\!^-) \leq 0$, it suffice to demonstrate that $(1-\beta) \big[V(b-1,\Delta\!^+) - V(b-1,\Delta\!^-) \big] - (1-p_t) \big[ V(b,\Delta\!^+) - V(b,\Delta\!^-) \big] \leq 0$, for $b>0$ and $1 < \Delta\!^- \leq \Delta\!^+$.
		We utilize the VIA and mathematical induction to proceed with the proof. VIA converges to the value function of Bellman's equation irrespective of the initial value assigned to $V_0(s)$, i.e., $\lim_{k\rightarrow\infty}{V_k(s)}=V(s)\ \forall s\in S$. Therefore, it suffices to establish the following inequality for all $k \in \{0,1,2,\cdots\}$:
		\begin{multline}
			\label{eqn:DV_diff_iter_k}
			(1-\beta) \big[V_k(b\!-\!1,\Delta\!^+) \!-\! V_k(b\!-\!1,\Delta\!^-) \big] \\
			\!-\! (1\!-\!p_t) \big[ V_k(b,\Delta\!^+) \!-\! V_k(b,\Delta\!^-) \big] \leq 0.
		\end{multline}
		
		Assuming $V_0(s) = 0$ for all $s \in S$, \eqref{eqn:DV_diff_iter_k} holds true for $k=0$. Now, with the same assumption extending up to $k > 0$, we prove its validity for $k+1$.
		By defining $V^0_{k+1}(s)=\sum_{s^\prime\in S}P\left[s^\prime|s,a=0\right]V_k(s^\prime)$, $V^1_{k+1}(s)\!=\!\sum_{s^\prime\in S}\!P\left[s^\prime|s,a\!=\!1\right]V_k(s^\prime)$, 
		
		\begin{subequations}
			\label{eqn:RVI_VoV1}
			\begin{align}
				V&^0_{k+1}(s) = \alpha p_t V_k(b\!+\!1,0) + \alpha \left(1\!-\!p_t\right) V_k(b\!+\!1,\Delta\!+\!1) \notag \\ 
				& + (1\!-\!\alpha) p_t V_k(b,0) + (1\!-\!\alpha) \left(1\!-\!p_t\right) V_k(b,\Delta\!+\!1), \\
				V&^1_{k+1}(s) = \alpha \beta V_k(b,0) + \alpha  \left(1\!-\!\beta\right) V_k(b,\Delta\!+\!1) \\
				& + (1\!-\!\alpha) \beta V_k(b\!-\!1,0) + (1\!-\!\alpha) \left(1\!-\!\beta\right) V_k(b\!-\!1,\Delta\!+\!1), \notag
			\end{align}
		\end{subequations}
		
		\noindent the VIA equation is given by $V_{k+1}(b,\Delta) = \Delta + \min \{V^0_{k+1}(b,\Delta),V^1_{k+1}(b,\Delta)\}$; thus the inequality \eqref{eqn:DV_diff_iter_k} for $k+1$ can further be simplified:
		\begin{align}
			\label{eqn:DV_diff_iter_k+1_sim}
			& \underbrace{(p_t-\beta)}_{< 0}\underbrace{(\Delta\!^+ - \Delta\!^-)}_{\geq 0} \notag \\
			&+ (1-\beta) \Big[ \min \{V^0_{k+1}(b-1,\Delta\!^+),V^1_{k+1}(b-1,\Delta\!^+)\} \notag \\
			&\qquad \qquad - \min \{V^0_{k+1}(b-1,\Delta\!^-),V^1_{k+1}(b-1,\Delta\!^-)\} \Big] \notag \\
			&- (1-p_t) \Big[ \min \{V^0_{k+1}(b,\Delta\!^+),V^1_{k+1}(b,\Delta\!^+)\} \notag \\
			&\qquad \qquad - \min \{V^0_{k+1}(b,\Delta\!^-),V^1_{k+1}(b,\Delta\!^-)\} \Big] \leq 0.
		\end{align}
		
		The first term in \eqref{eqn:DV_diff_iter_k+1_sim} is non-positive; thus, it is sufficient to demonstrate that the other terms, which we denote by $E_1$, are non-positive, i.e., $E_1 \leq 0$. To proceed with the proof, we consider four cases. 
		Case 1, where $V^0_{k+1}(b-1,\Delta\!^-) \leq V^1_{k+1}(b-1,\Delta\!^-)$ and $V^0_{k+1}(b,\Delta\!^+) \leq V^1_{k+1}(b,\Delta\!^+)$;
		case 2, where $V^0_{k+1}(b-1,\Delta\!^-) \leq V^1_{k+1}(b-1,\Delta\!^-)$ and $V^0_{k+1}(b,\Delta\!^+) > V^1_{k+1}(b,\Delta\!^+)$;
		cases 3 and 4 are defined by reversing the inequality signs in cases 1 and 2.
		We prove the inequality $E_1\leq 0$ for case 1; a similar approach can be utilized to prove the other cases. In this case, equation $E_1\leq 0$ is simplified:
		
		{\small
			\begin{align}
				&(1\!-\!\beta) \Big[ \min \{V^0_{k+1}(b\!-\!1,\Delta\!^+),V^1_{k+1}(b\!-\!1,\Delta\!^+)\} \!-\! V^0_{k+1}(b\!-\!1,\Delta\!^-) \Big] \notag \\
				&\quad - (1\!-\!p_t) \Big[ V^0_{k+1}(b,\Delta\!^+) \!-\! \min \{V^0_{k+1}(b,\Delta\!^-),V^1_{k+1}(b,\Delta\!^-)\} \Big] \leq 0 \notag \\
				& \underset{(a)}{\Leftrightarrow} (1\!-\!\beta) \Big[ V^0_{k+1}(b\!-\!1,\Delta\!^+) \!-\! V^0_{k+1}(b\!-\!1,\Delta\!^-) \Big] \! \notag \\
				&\qquad +\! (1\!-\!\beta) \min \{0,V^1_{k+1}(b\!-\!1,\Delta\!^+)\!-\!V^0_{k+1}(b\!-\!1,\Delta\!^+)\}  \notag \\
				&\qquad - (1\!-\!p_t) \Big[ V^0_{k+1}(b,\Delta\!^+) \!-\! V^0_{k+1}(b,\Delta\!^-)\Big] \notag \\
				&\qquad + (1\!-\!p_t) \min \{0,V^1_{k+1}(b,\Delta\!^-)\!-\!V^0_{k+1}(b,\Delta\!^-)\}   \leq 0,
			\end{align}
		}
		
		\noindent where $(a)$ is resulted from $\min\left\{x,y\right\}=x+\min\left\{0,y-x\right\}$. Since the second and last terms are negative (non-positive), it suffices to show that: 
		\begin{multline}
			(1-\beta) \Big[ V^0_{k+1}(b-1,\Delta\!^+) - V^0_{k+1}(b-1,\Delta\!^-) \Big] \\
			- (1-p_t) \Big[ V^0_{k+1}(b,\Delta\!^+) - V^0_{k+1}(b,\Delta\!^-)\Big]    \leq 0,
		\end{multline}
		where, according to \eqref{eqn:RVI_VoV1} and after some manipulation, it can be expressed as follows:
		
		\begin{align}
			& \alpha \left(1\!-\!p_t\right)  \Big\{ (1\!-\!\beta) \big[V_k(b,\Delta\!^+\!+\!1) \!-\! V_k(b,\Delta\!^-\!+\!1) \big] \notag \\
			& -\! (1\!-\!p_t) \big[V_k(b\!+\!1,\Delta\!^+\!+\!1) \!-\! V_k(b\!+\!1,\Delta\!^-\!+\!1) \big] \Big\} \notag \\
			& +\! (1\!-\!\alpha) \! \left(1\!-\!p_t\right) \! \Big\{\! (1\!-\!\beta) \big[\!V_k(b\!-\!1,\Delta\!^+\!+\!1) \!-\! V_k(b\!-\!1,\Delta\!^-\!+\!1)\big] \notag \\
			& -\! (1\!-\!p_t) \big[V_k(b,\Delta\!^+\!+\!1) \!-\! V_k(b,\Delta\!^-\!+\!1)\big] \Big\} \leq 0,
		\end{align}
		where both the expressions within the braces are negative according to \eqref{eqn:DV_diff_iter_k}, and the proof is complete.
	\end{proof}

\end{document}